\documentclass[journal]{IEEEtran}

\usepackage{amsmath}
\usepackage{amsfonts}
\usepackage{amssymb}
\usepackage{amsthm}
\usepackage{tabularx}
\usepackage{mathrsfs} 

\usepackage{url}
\usepackage{hyperref}
\newtheorem{theorem}{Theorem}

\usepackage{stfloats}
\usepackage{float}
\usepackage{graphicx}
\usepackage{cite}
\usepackage{xcolor}
\usepackage{subfigure}
\renewcommand{\vec}[1]{\mathbf{#1}}

\makeatletter
\def\blfootnote{\xdef\@thefnmark{}\@footnotetext}
\makeatother

\begin{document}
	
	\title{\LARGE{RIS-Aided Communications over Dirty MAC:\\ Capacity Region and Outage Probability}} 
	\author{Farshad~Rostami~Ghadi\IEEEmembership{} and  F.~Javier~L\'opez-Mart\'inez\IEEEmembership{}}
	\maketitle
	\begin{abstract}
		%{%\textcolor{blue}
		%}
	We characterize the capacity region of a two-user multiple access channel (MAC) in a reconfigurable intelligent surface (RIS)-aided communication system with side information (SI) at the transmitters. We consider two uplink communication scenarios: (i) a double dirty MAC where the interferences are known non-causally to both users; and (ii) a single dirty MAC model where only one of the users knows the interference and the other one is not aware of the interference. Considering that each of the users is assisted by a different RIS with $M_i$ elements, we derive closed-form expressions for the capacity region and the outage probability (OP) for both system models. Results show that the use of RISs is beneficial to extend the capacity region and improve the OP in both scenarios.

	\end{abstract}
	%\vspace{-3.5ex}
	\maketitle
	\blfootnote{\noindent Manuscript received XX August, 2022. This work was funded in part by Junta de Andaluc\'ia, the European Union and the European Fund for Regional Development FEDER through grants P18-RT-3175 and EMERGIA20-00297, and in part by MCIN/AEI/10.13039/501100011033 through grant PID2020-118139RB-I00. The review of this paper was coordinated by XXXX.%This work has been funded in part by the Spanish Government and the European Fund for Regional Development FEDER (project  TEC2017-87913-R) and by Junta de Andalucia (project P18-RT-3175). 
	}
	%
	%\blfootnote{\noindent  and G.A. Hodtani are with Department of Electrical Engineering, Ferdowsi University of Mashhad, Mashhad, Iran. (e-mail: $\rm \{f.rostami.gh,ghodtani\}@gmail.com$).}
	
	\blfootnote{\noindent The authors are with the Communications and Signal Processing Lab, Telecommunication Research Institute (TELMA), Universidad de M{\'a}laga, M{\'a}laga, 29010, (Spain). F.J. L{\'o}pez Mart{\'i}nez is also with the Dept. Signal Theory, Networking and Communications, University of Granada, 18071, Granada (Spain). (e-mail: $\rm farshad@ic.uma.es$, $\rm fjlm@ugr.es$).}
	
	\blfootnote{Digital Object Identifier 10.1109/XXX.2021.XXXXXXX}
	%\IEEEpeerreviewmaketitle
	\vspace{-3mm}
	\section{Introduction}\label{introduction}
	Reconfigurable intelligent surfaces (RISs) have been recently introduced as a promising approach for future wireless communication networks due to their capability to enhance coverage area and provide high spectral/energy efficiency \cite{basar2019wireless}. %through adaptively controlling the wireless signal propagation environment.
	 Specifically, an RIS is an artificial metasurface with a large number of low-cost passive reflecting elements, where each element can independently introduce a phase shift on the reflected signal. By smartly adjusting the phase shifts induced by all the reflecting elements, the RIS can adaptively control the wireless signal propagation environment to maximize the desired signal quality. Therefore, the achievable rate and reliability of the corresponding transmission, which are momentous challenges in designing sixth-generation (6G) technology, are improved. On the other hand, the use of side information (SI) at the transmitter has the potential to help meet the reliability constraint in future wireless networks, since such knowledge (e.g., either channel state information (CSI) or interference awareness) can be leveraged to intelligently encode their information. Considering SI at the transmitters in multi-user systems can reduce the destructive effects of the interference and provide reliable communication at higher rates \cite{jafar2006capacity,philosof2011lattice,ghadi2021role}. Hence, investigation of the RIS-aided communication systems in the presence of SI at transmitters can be worthy of attention.
	 
	 In the last years, intense research activities have been mostly focused on the phase shift design and performance analysis of RIS in various single-user \cite{ye2020joint,cui2019secure,ai2021secure,yang2020secrecy} and multi-user \cite{mu2021joint,zhou2020intelligent,wei2021channel,gan2021ris,xiao2021average,liu2021dynamic} scenarios. On the contrary, information-theoretic aspects of RIS communications have been scarcely addressed in the literature. In this regard, the capacity characterization for RIS-aided multiple-input multiple-output (MIMO) systems in a single-user setup was investigated in \cite{zhang2020capacity1,karasik2020beyond}, while the \textit{capacity region} for more complex RIS-aided multi-user systems was only studied in \cite{zhang2021intelligent}. The authors in \cite{zhang2021intelligent} considered a two-user \textit{clean}\footnote{i.e., without side information at the transmitters} multiple access channel (MAC), where users send independent messages to a common receiver aided by reflecting elements of RIS. By proposing two strategies for RIS location, they provided capacity and achievable rate regions with practical orthogonal multiple access (OMA) schemes such as time-division multiple access (TDMA) and frequency-division multiple access (FDMA), and showing that RISs significantly improve the system performance. 

To the best of the authors' knowledge, the role of SI in RIS-assisted multiple access communications remains an open challenge. For this purpose and motivated by the potential of RIS and SI approaches in providing more reliable communications for the next generation of wireless networks, we extend the RIS-aided two-user clean MAC considered in \cite{zhang2021intelligent} to the RIS-aided two-user \textit{dirty} MAC model by assuming known SI at the transmitters. In particular, we consider two uplink communication scenarios under RIS deployment, i.e., a doubly dirty MAC where the interferences are known non-causally to both users; and a single dirty MAC model where only one of the users knows the interference and the other one is not aware of the interference. By considering two separate RIS with $M_i$ reflecting elements near to each user, we derive the capacity region for both doubly and single dirty MAC models, and derive closed-form expressions for the outage probability.\vspace{-0.4cm}
	 %Full exploiting all benefits of the RIS in practical applications of wireless networks needs to meet new challenges in this area. In this regard, many works have been recently done related to the performance of RIS in single-user and multi-user wireless communication systems from different aspects.
%	\begin{figure}[!h]\vspace{0ex}
%	\centering
%	\includegraphics[width=0.6\columnwidth]{ddmac.jpg} %\textwidth
%	\caption{Capacity region for doubly dirty MAC under RIS deployment when $P_1=P_2=30$dB, $N=-40$dB, and $\check{d}_1=\check{d}_2=20$m.} %\vspace{-0.6cm}
%	\label{system-dd}
%\end{figure}%\vspace{0ex}
%\begin{figure}[!h]\vspace{0ex}
%	\centering
%	\includegraphics[width=0.6\columnwidth]{sdmac.jpg} %\textwidth
%	\caption{Capacity region for doubly dirty MAC under RIS deployment when $P_1=P_2=30$dB, $N=-40$dB, and $\check{d}_1=\check{d}_2=20$m.} %\vspace{-0.6cm}
%	\label{system-sd}
%\end{figure}%\vspace{0ex}
%\begin{figure}[!t]
%	\centering
%	\hspace{0cm}\subfigure[Doubly dirty MAC]{%
%		\includegraphics[width=0.25\textwidth]{ddmac.jpg}%
%		\label{system-ddmac}%
%	}\hspace{0.8cm}%or more
%	\subfigure[Single dirty MAC]{%
%		\includegraphics[width=0.25\textwidth]{sdmac.jpg}%
%		\label{system-sdmac}%
%	}\hspace{0cm}%or more
%	\caption{RIS-aided two-user multiple access communication systems: (a) doubly dirty MAC and (b) single dirty MAC.}\vspace{0cm}
%	\label{system}
%\end{figure}	
	\section{System Model}\label{sec-model}
\subsection{Doubly dirty MAC}
We consider a wireless two-user MAC with %non-causally 
known interferences $S_i$, $i\in\{1,2\}$ %at the transmitters (i.e., doubly dirty MAC) 
as shown in Fig. \ref{system-model}(a), where two single-antenna users $u_i$ that are sufficiently far apart from each other aim to send independent inputs $X_i$ to a common single-antenna receiver $r$, respectively. We assume that the interference signals $S_i$ with variances $Q_i$ ($S_i\sim\mathcal{N}(0,Q_i)$) are known non-causally to the transmitters of users $u_i$, respectively, and the inputs $X_i$ sent by users $u_i$ over the channels are subjected to the average power constraint as $\mathbb{E}\left[|X_i|^2\right]\leq P_i$, respectively. Furthermore, in order to improve communication rates, we consider two RISs that contain $M_i\ge 1$ passive reflecting elements in the vicinity of user $u_i$. Each RIS element is designed to be able to induce an independent phase shift to the incident signal for collaboratively altering the effective channels between the users and the receiver. Given the assumption that the users are sufficiently far apart, we suppose that the signal transmitted by user $u_i$ and reflected by the serving RIS to user $u_j$, $i\neq j$, is negligible at the receiver $r$ due to the severe path-loss. Therefore, the equivalent channel observed by receiver $r$ from the $i$-th user includes the direct link from user $u_i$ and the reflected link by its serving RIS. We also assume perfect CSI availability so that RISs phase shifts are optimally configured. Under such assumptions, the received signal $Y$ at the receiver $r$ can be defined as:
 \begin{align}
 	Y&=\sum_{i=1}^{2}\left[\tilde{h}_i(\vec{\Phi}_i)X_i+S_i\right]+Z,\\
 	&=\sum_{i=1}^{2}\left[\left(\hat{h}_i+\vec{g}^T_i\vec{\Phi}_i\vec{h}_i\right)X_i+S_i\right]+Z,\label{eq-y1}
 \end{align}
%\begin{align}
%	Y=\tilde{h}_1\left(\vec{\Phi}_1\right)X_1+\tilde{h}_2\left(\vec{\Phi}_2\right)X_2+S_1+S_2+Z
%\end{align}
where $\tilde{h}_i\left(\vec{\Phi}_i\right)$ and $\hat{h}_i\in\mathbb{R}$ denote the effective channel and the baseband direct channel from user $u_i$ to the receiver $r$, respectively. The terms $Z$ defines the Additive White Gaussian Noise
(AWGN) with zero mean and variance $N (Z\sim\mathcal{CN}(0,N))$ at the receiver $r$ and $\vec{\Phi}_i\in\mathbb{R}^{M_i\times M_i}$ is the adjustable response induced by the $m$-th reflecting meta-surface of the $i$-th RIS. For the sake of simplicity, we assume that the RISs do not induce attenuation to the signal i.e., $|\phi_{im}|=1, \forall m\in \mathcal{M}_i\in\left\{1,...,M_i\right\}$, which is defined as:
\begin{align}
\vec{\Phi}_i=\text{diag}\left(\left[\mathrm{e}^{j\varphi_{i1}}, \mathrm{e}^{j\varphi_{i2}},...,\mathrm{e}^{j\varphi_{iM_i}}\right]\right).\label{phii}
\end{align}
The vector $\vec{h}_i\in\mathbb{R}^{M_i\times 1}$ contains the channel gains from user $u_i$
to each element of its serving (nearby) RIS and vector  $\vec{g}^T_i\in\mathbb{R}^{1\times M_i}$ includes the channel gains form each element of its serving RIS to the common receiver $r$, which are given by
$
\vec{h}_i=\left[h_{i1}\mathrm{e}^{-j\theta_{i1}},h_{i2}\mathrm{e}^{-j\theta_{i2}},...,h_{iM_i}\mathrm{e}^{-j\theta_{iM_i}}\right]^T%\label{hi}
$ and $
	\vec{g}_i=\left[g_{i1}\mathrm{e}^{-j\psi{i1}},g_{i2}\mathrm{e}^{-j\psi_{i2}},...,g_{iM_i}\mathrm{e}^{-j\psi_{iM_i}}\right]^T,%\label{gi}
	$
where $h_{iM_i}$ and $g_{iM_i}$ are the amplitudes of the corresponding channel gains, and $\mathrm{e}^{-j\theta_{iM_i}}$ and $\mathrm{e}^{-j\psi_{iM_i}}$ denote the phase of the corresponding links.\vspace{-0.5cm}
\begin{figure}[!t]\hspace{-0.6cm}
	\centering
	\includegraphics[width=1.05\columnwidth]{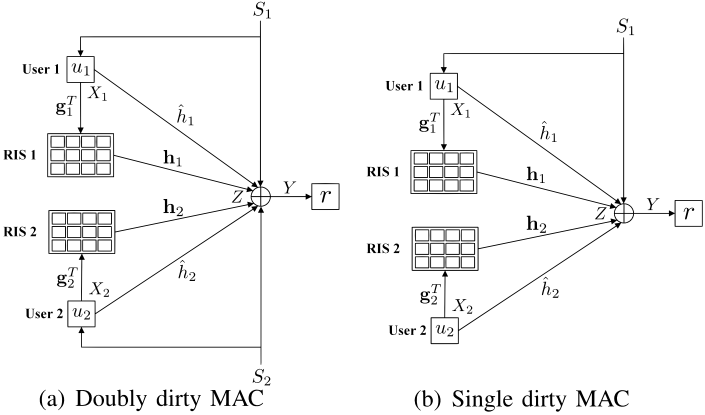} \vspace{-0.5cm}%\textwidth
	\caption{RIS-aided two-user multiple access communication systems.}\vspace{-0.6cm}%: (a) doubly dirty MAC and (b) single dirty MAC.}\vspace{0cm} %\vspace{-0.6cm}
	\label{system-model}
\end{figure}
\subsection{Single dirty MAC}
Here, we consider a wireless MAC that includes a single dirty user and helper problem as shown in Fig. \ref{system-model}(b), where only user $u_1$ knows the interference $S_1$ and user $u_2$ has no awareness of the interference. In this asymmetric scenario, user $u_2$ aims to send the message to the receiver $r$ at its maximum rate with the help of user $u_1$, where user $u_1$ does not send any message in this case. Thus, for this considered model, the received signal $Y$ at the receiver $r$ can be expressed as:
 \begin{align}
	Y&=\sum_{i=1}^{2}\tilde{h}_i(\vec{\Phi}_i)X_i+S_1+Z,\\
	&=\sum_{i=1}^{2}\left(\hat{h}_i+\vec{g}^T_i\vec{\Phi}_i\vec{h}_i\right)X_i+S_1+Z.
\end{align}
%where $\tilde{h}_2\left(\vec{\Phi}_2\right)$ and $\hat{h}_2\in\mathbb{R}$ are the effective channel and the baseband direct channel from user $u_2$ to the receiver $r$, respectively. The terms $\vec{\Phi}_2$, $\vec{h}_2$, and $\vec{g}^T_2$ are given in \eqref{phii}, \eqref{hi}, and \eqref{gi} for $i=2$, respectively.

%In the following, we characterize the capacity region of doubly/single dirty MAC in the presence of strong interferences under RIS consideration, which constitutes the achievable rate-pair $\left(R_1,R_2\right)$.\vspace{-1cm}
\section{Capacity Region}\label{sec-capacity}
In this section, we characterize the capacity region of doubly/single dirty MAC in the presence of strong interferences under RIS consideration, which constitutes the achievable rate-pair $\left(R_1,R_2\right)$.\vspace{-0.4cm}
 %we provide the capacity region for both doubly and single dirty MAC models under strong interference scenario and RIS deployment.
\subsection{Doubly Dirty MAC}
Note that under RIS reflection coefficients (i.e., phase shifts) $\{\vec{\Phi}_i\}$, the equivalent channels observed by receiver $r$ from users $u_i$ over doubly dirty MAC are determined by $\{\tilde{h}_i\left(\vec{\Phi}_i\right)\}$ as per \eqref{eq-y1}, and the instantaneous capacity region $C\left(\left\{\vec{\Phi}_i\right\}\right)$ of the two-user doubly dirty MAC with coherent receiver (fading coefficients $\{\tilde{h}_i\left(\vec{\Phi}_i\right)\}$ are known
at the receiver) under strong interference condition (i.e., $Q_i\rightarrow\infty$) is a triangle region \cite{philosof2011lattice} including rate pairs that satisfy the following constraint:
\begin{align}
 \hspace{-1ex}R_1+R_2\leq\log_2\left(1+\min\left\{\frac{P_1|\tilde{h}_1\left(\vec{\Phi}_1\right)|^2}{N},\frac{P_2|\tilde{h}_2\left(\vec{\Phi}_2\right)|^2}{N}\right\}\right),\label{region-dd-without}
\end{align}
where by flexibly modeling the RIS reflection coefficients $\{\vec{\Phi}_i\}$, any rate pair within the union set of $C^\mathsf{D}\left(\left\{\vec{\Phi}_i\right\}\right)$ over all possible $\{\vec{\Phi}_i\}$ can be
obtained. Therefore, the capacity region of doubly dirty MAC under strong interference scenario and RIS deployment is defined as the convex hull of such a union set:
\begin{align}
C^\mathsf{D}\overset{\Delta}{=}\text{Conv}\left(\bigcup\underset{\left\{\vec{\Phi}_i\right\}\in\mathcal{A}}{}C^\mathsf{D}\left(\left\{\vec{\Phi}_i\right\}\right)\right),\label{union}
\end{align}
where $\mathcal{A}\overset{\Delta}{=}\left\{\left\{\vec{\Phi}\right\}:|\phi_{im}|=1, \forall i, m\right\}$ is the feasible set of $\left\{\vec{\Phi}_i\right\}$. 

In order to determine the closed-form expression of the capacity region for the considered doubly dirty MAC, we first define the effective channel gain for
each user $u_i$ as follows: 
\begin{align}
|\tilde{h}_i\left(\vec{\Phi}_i\right)|&=|\hat{h}_i+\sum_{m=1}^{M_i}g_{im}\phi_{im}h_{im}|,\\ \label{heq1}
&=\Big|\hat{h}_i+\sum_{m=1}^{M_i}g_{im}h_{im}\mathrm{e}^{j\left(\varphi_{im}-\theta_{im}-\psi_{im}\right)}\Big|,\\\label{phi_i}
&\overset{(a)}{=}|\hat{h}_i|+\Big|\sum_{m=1}^{M_i}g_{im}h_{im}\Big|\overset{\Delta}{=}\tilde{h}_{i,\mathrm{V}},
\end{align}
where $(a)$ follows from assuming perfect CSI for RIS configuration, which enables ideal phase shifting (i.e., $\varphi_{im}=\theta_{im}+\psi_{im}$). Now, we present the following theorem that determine the closed-from capacity region of doubly dirty MAC under RIS deployment. 
\begin{theorem}\label{thm-dd}
The capacity region of the RIS-aided doubly dirty MAC under strong interference scenario is given by
\begin{align}
	C^\mathsf{D}=\left\{\left(R_1,R_2\right):R_1+R_2\leq R^{\mathsf{D}}_{1,2}\right\},\label{region-dd}
\end{align}
where 
$
R^{\mathsf{D}}_{1,2}\overset{\Delta}{=}\log_2 \left(1+\min\left\{\frac{P_1\tilde{h}_{1,\mathrm{V}}}{N},\frac{P_2\tilde{h}_{2,\mathrm{V}}}{N}\right\}\right).
$
\end{theorem}
\begin{proof}\label{proof-th1}
By noting that the capacity region in \eqref{region-dd} is an achievable rate region with $\left\{\vec{\Phi}_i\right\}$ in \eqref{phi_i} and also provides a convex outer bound for all achievable $C^\mathsf{D}\left(\left\{\vec{\Phi}_i\right\}\right)$ in \eqref{region-dd-without}, the proof is completed. \vspace{-0.5cm}
\end{proof}
\subsection{Single Dirty MAC}
Here, we provide the closed-form capacity region of single dirty MAC under RIS deployment. It is worth noting that since we consider the strong interference scenario, the MAC with a single (strong) dirty user is \textit{not} an specific case of the doubly dirty MAC with strong interferences \cite{philosof2011lattice}. By considering RIS reflection coefficient $\{\vec{\Phi}_i\}$, the effective channel observed by receiver $r$ from user $u_2$ over single dirty MAC is given by $\tilde{h}_2\left(\vec{\Phi}_2\right)$, and the instantaneous capacity region $C^\mathsf{S}\left(\{\vec{\Phi}_i\}\right)$ of the single dirty MAC with coherent receiver under strong interference case (i.e., $Q_1\rightarrow \infty$) is a triangle/quadrilateral region consisting of rate pairs that satisfy the following constraints: 
\begin{align}
	\hspace{-2.4ex}R_2\leq\log_2\left(1+\min\left\{\frac{P_1|\tilde{h}_1\left(\vec{\Phi}_1\right)|^2}{N},\frac{P_2|\tilde{h}_2\left(\vec{\Phi}_2\right)|^2}{N}\right\}\right),\label{rate1}
\end{align}\vspace{-0.4cm}
\begin{align}
	R_2+R_2\leq\log_2\left(\frac{P_1|\tilde{h}_1\left(\vec{\Phi}_1\right)|^2}{N}\right),
\label{rate2}
\end{align}
where similar to the doubly dirty MAC, any rate pair within the union set of $C^\mathsf{S}\left(\left\{\vec{\Phi}_i\right\}\right)$ over all feasible $\{\vec{\Phi}_i\}$ can be achieved by flexibly designing the RIS reflection coefficients $\left\{\vec{\Phi}_i\right\}$. Thus, the capacity region of single dirty MAC under strong interference case and RIS deployment is defined as the convex hull of such a union set given in \eqref{union}. Now, exploiting the effective channel gain for each user $u_i$ given in \eqref{heq1}, and setting the phase shift as $\varphi_{im}=\theta_{im}+\psi_{im}$, we present the following theorem to determine the closed-form capacity region of single dirty MAC under RIS deployment.
\begin{theorem}
The capacity region of the RIS-aided single dirty MAC under strong interference scenario is given by
\begin{align}\label{region-sd}
C^\mathsf{S}=\left\{\left(R_1,R_2\right): R_2\leq R^{\mathsf{S}}_2, R_1+R_2\leq R^{\mathsf{S}}_{1,2}\right\},
\end{align}
where 
$
	R^{\mathsf{S}}_{2}\overset{\Delta}{=}\log_2 \left(1+\min\left\{\frac{P_1\tilde{h}_{1,\mathrm{V}}}{N},\frac{P_2\tilde{h}_{2,\mathrm{V}}}{N}\right\}\right) \text{and} \,
	R^{\mathsf{S}}_{1,2}=\log_2\left(1+\frac{P_1\tilde{h}_{1,\mathrm{V}}}{N}\right).
$
\end{theorem}
\begin{proof}
Similar to the proof of Theorem \ref{thm-dd}, by noting the fact that the capacity region in \eqref{region-sd} provides a convex outer bound for all achievable $C^\mathsf{S}\left(\left\{\vec{\Phi}_i\right\}\right)$ in \eqref{rate1} and \eqref{rate2}, the proof is completed.\vspace{-0.5cm}
\end{proof}
\section{Outage Probability}
In this section, by exploiting the capacity regions that we obtained in Section \ref{sec-capacity}, we derive the closed-from expression of the OP for both doubly and single dirty MAC models under RIS deployment. 

OP is a typical performance measure for communication systems operating over fading channels and it is defined as the probability that the channel capacity is less than a certain information rate $R_t>0$: 
$P_o\overset{\Delta}{=}\Pr\left(C\leq R_t\right).$
In order to obtain the OP for the considered system models, we first need to characterize the signal-to-noise ratio (SNR) distributions. To this end, we assume that all the channels undergo Rayleigh fading and the channel gains follow the complex Gaussian distribution with zero mean and unit variance.  We also suppose that $\hat{d}_{i}$ defines the distance between user $u_i$ and the receiver $r$, $\bar{d}_i$ represents the distance between user $u_i$ and its serving RIS, and $\tilde{d}_i$ denotes the distance between the serving RIS to the $i$-th user and the receiver $r$. Thus, by considering the effective channel after setting ideal shift phase (i.e., $\tilde{h}_{i,\mathrm{V}}$), the instantaneous SNRs at the receiver $r$ from users $u_i$ can be expressed as:\vspace{-0.2cm}
\begin{align}
\hspace{-0.2cm}\gamma_i\approx\frac{P_i|\hat{h}_i|^2}{\hat{d}_i^{\hat{\alpha}}N}+\frac{P_i\Big|\sum_{m=1}^{M_i}g_{im}h_{im}\Big|^2}{\bar{d}_i^{\bar{\alpha}}\tilde{d}_i^{\tilde{\alpha}}N}=\hat{\gamma}_i|\hat{h}_i|^2+\tilde{\gamma}_iH_i^2,
\end{align}
where $\hat{\alpha}$, $\bar{\alpha}$, and $\tilde{\alpha}$ denote the corresponding path-loss exponents, and $\hat{\gamma}_i$ and $\tilde{\gamma}_i$ are the average SNRs associated to the direct and through-RIS paths, respectively. Since $g_{im}$ and $h_{im}$ are independent Rayleigh random variables (RVs), the mean value and the variance of their product are $\mathbb{E}(g_{im}h_{im})=\pi/4$ and $\text{Var}(g_{im}h_{im})=1-\pi^2/16$, respectively. According to the central limit theorem (CLT), $H_i$ converges to a Gaussian distributed RV with mean value $\mathbb{E}(H_i)=M_i\pi/4$ and variance $\text{Var}(H_i)=M_i\left(1-\pi^2/16\right)$ for a sufficiently large number of reflecting meta-surfaces, i.e., $M_i\gg 1$. Thus, $H_i^2$ is a non-central chi-square RV with one degree of freedom and has the probability density function (PDF) as 
$
f_{H_i^2}(h_i^2)=\frac{\lambda_i}{\sqrt[4]{h_i^2}}\mathrm{e}^{-\zeta_i h_i^2}I_{-\frac{1}{2}}\left(\sigma \sqrt{h_i^2}\right),
$
where $\lambda_i=\frac{4\sqrt{\pi M_i}\mathrm{e}^{-\frac{M_i\pi^2}{2\left(16-\pi^2\right)}}}{M_i(16-\pi^2)}$, $\zeta_i=\frac{8h_i^2}{M_i(16-\pi^2)}$, $\sigma=\frac{4\pi}{16-\pi^2}$, and $I_{\zeta}(.)$ denotes the modified Bessel function of the first kind \cite{proakis2001digital}. However, for ease of calculations in deriving the OP, we exploit the flexible approach provided in \cite{atapattu2011mixture} to express the PDF of $H_i^2$. Hence, $f_{H_i^2}(h_i^2)$ can be re-written as:\vspace{-0.2cm}
\begin{align}
f_{H_i^2}(h_i^2)=\sum_{l=1}^{L}\beta_{il}x^{\kappa_l-1}\mathrm{e}^{-\zeta_i h_i^2},\label{eq-pdf-H}
\end{align}
where $\beta_{il}=\frac{\lambda_i\left(\sigma/2\right)^{2l-5/2}}{(l-1)!\Gamma(l-1/2)}$, $\kappa_l=l-1/2$, $L$ is the number of terms, and $\Gamma(.)$ denotes the gamma function. So, the cumulative distribution function (CDF) of $\gamma_i$ can be defined as:
\begin{align}
&F_{\gamma_i}(\gamma_i)=\Pr\left(\hat{\gamma}_i|\hat{h}_i|^2+\tilde{\gamma}_iH_i^2\leq \gamma_i\right),\\
&=\int_{0}^{\infty}\int_{0}^{\frac{\gamma_i-\bar{\gamma}_ih_i^2}{\hat{h}_i}}f_{|\hat{h}_i|^2}(|\hat{h}_i|^2)f_{h_i^2}(h_i^2)d|\hat{h}_i|^2dh_i^2.\label{eq-int-cdf}
\end{align}
Now, by considering Rayleigh distribution for $|\hat{h}_i^2|$, inserting \eqref{eq-pdf-H} into \eqref{eq-int-cdf}, and computing the integral, $F_{\gamma_i}(\gamma_i)$ can be determined as \cite{yang2020secrecy}:
\begin{align}
F_{\gamma_i}(\gamma_i)=\sum_{l=1}^{L}\frac{\beta_{il}}{\zeta_i^{\kappa_l}}\gamma\left(\kappa_l,\frac{\gamma_i\zeta_i}{\tilde{\gamma}_i}\right)-\mathrm{e}^{-\frac{\gamma_i}{\hat{\gamma}_i}}\sum_{l=1}^{L}\frac{\beta_{il}\tilde{\gamma}_i^{\kappa_l}}{\omega_i^{\kappa_l}}\gamma\left(\kappa_l,\omega_i\gamma_i\right),\label{cdf}
\end{align}
where $\omega_i=\frac{\zeta_i\left(\hat{\gamma}_i-\tilde{\gamma}_i\right)}{\hat{\gamma}_i\tilde{\gamma}_i}$, and $\gamma(.,.)$ is the lower incomplete gamma function. By definition, the PDF of $\gamma_i$ can be achieved as:\vspace{-0.2cm}
\begin{align}
	f_{\gamma_i}(\gamma_i)=\sum_{l=1}^{L}\beta_{il}\left(\frac{\gamma_i}{\tilde{\gamma}_i}\right)^{\kappa_l-1}\left(\mathrm{e}^{-\frac{\zeta_i\gamma_i}{\tilde{\gamma}_i}}-\mathrm{e}^{-\omega_i\gamma_i}\right).
\end{align}\vspace{-0.8cm}
\subsection{Doubly Dirty MAC}
Now, by exploiting the capacity region that given in \eqref{region-dd}, we define the OP for doubly dirty MAC under RIS deployment as:\vspace{-0.2cm}
\begin{align}
&P_o^\mathsf{D}=\Pr\left(C^\mathsf{D}\leq R_t^{\mathsf{D}}\right),\\
&=\Pr\left(\log_2 \left(1+\min\left\{\gamma_1,\gamma_2\right\}\right)\leq R_t^{\mathsf{D}}\right),\\
&=1-\bar{F}_{\gamma_1}\left(\gamma_t^\mathsf{D}\right)\bar{F}_{\gamma_2}\left(\gamma_t^\mathsf{D}\right),\label{eq-out-a}
\end{align}
where $\bar{F}_{\gamma_i}(\gamma_t^\mathsf{D})=1-{F}_{\gamma_i}(\gamma_t^\mathsf{D})$ is the complementary CDF of $\gamma_i$ and $\gamma_t^\mathsf{D}=2^{R_t^\mathsf{D}}-1$. Consequently, by  plugging $F_{\gamma_i}(\gamma_t^\mathsf{D})$ from \eqref{cdf} into \eqref{eq-out-a}, the OP for doubly dirty MAC is determined as \eqref{eq-out-d}. \vspace{-0.5cm}
		\begin{figure*}[t]
	\normalsize
	%\hrulefill
	\setcounter{equation}{23}
	\small\begin{align}
	P_o^\mathsf{D}=&\,1-\left(1-\sum_{l=1}^{L}\frac{\beta_{1l}}{\zeta_1^{\kappa_l}}\gamma\left(\kappa_l,\frac{\gamma_t^\mathsf{D}\zeta_1}{\tilde{\gamma}_1}\right)+\mathrm{e}^{-\frac{\gamma_t^\mathsf{D}}{\hat{\gamma}_1}}\sum_{l=1}^{L}\frac{\beta_{1l}\tilde{\gamma}_1^{\kappa_l}}{\omega_1^{\kappa_l}}\gamma\left(\kappa_l,\omega_1\gamma_t^\mathsf{D}\right)\right) \left(1-\sum_{l=1}^{L}\frac{\beta_{2l}}{\zeta_2^{\kappa_l}}\gamma\left(\kappa_l,\frac{\gamma_t^\mathsf{D}\zeta_2}{\tilde{\gamma}_2}\right)+\mathrm{e}^{-\frac{\gamma_t^\mathsf{D}}{\hat{\gamma}_2}}\sum_{l=1}^{L}\frac{\beta_{2l}\tilde{\gamma}_2^{\kappa_l}}{\omega_2^{\kappa_l}}\gamma\left(\kappa_l,\omega_2\gamma_t^\mathsf{D}\right)\right).\label{eq-out-d}
	\end{align}
	\hrulefill\vspace{-0.6cm}
\end{figure*}
\subsection{Single Dirty MAC}
Here, by considering the capacity region provided in \eqref{region-sd}, we express the OP for single dirty MAC under RIS deployment as:
\begin{align}
P_o^\mathsf{S}=\rho P_{o1}^\mathsf{S}+(1-\rho) P_{o2}^\mathsf{S},\label{eq-out-sd}
\end{align}
where $0\leq\rho\leq 1$ is the event probability. $P_{o1}^\mathsf{S}$ can be determined as:
\begin{align}
P_{o1}^\mathsf{S}&=\Pr\left(C^\mathsf{S}\leq R_2^\mathsf{S}\right),\\
&=\Pr\left(\log_2 \left(1+\min\left\{\gamma_1,\gamma_2\right\}\right)\leq R_2^{\mathsf{S}}\right),\\
&=1-\bar{F}_{\gamma_1}\left(\gamma_2^\mathsf{S}\right)\bar{F}_{\gamma_2}\left(\gamma_2^\mathsf{S}\right),\label{eq-out1-sd}
\end{align}
where $\gamma_2^\mathsf{S}=2^{R_2^\mathsf{S}}-1$. Now, by inserting $F_{\gamma_i}(\gamma_2^\mathsf{S})$ into \eqref{eq-out1-sd}, $P_{o1}^\mathsf{S}$ is given by \eqref{eq-out-s1}. 
\begin{figure*}[t]
	\normalsize
	%\hrulefill
	\setcounter{equation}{28}
	\small\begin{align}
		P_{o1}^\mathsf{S}=&\,1-\left(1-\sum_{l=1}^{L}\frac{\beta_{1l}}{\zeta_1^{\kappa_l}}\gamma\left(\kappa_l,\frac{\gamma_2^\mathsf{S}\zeta_1}{\tilde{\gamma}_1}\right)+\mathrm{e}^{-\frac{\gamma_2^\mathsf{S}}{\hat{\gamma}_1}}\sum_{l=1}^{L}\frac{\beta_{1l}\tilde{\gamma}_1^{\kappa_l}}{\omega_1^{\kappa_l}}\gamma\left(\kappa_l,\omega_1\gamma_2^\mathsf{S}\right)\right) \left(1-\sum_{l=1}^{L}\frac{\beta_{2l}}{\zeta_2^{\kappa_l}}\gamma\left(\kappa_l,\frac{\gamma_2^\mathsf{S}\zeta_2}{\tilde{\gamma}_2}\right)+\mathrm{e}^{-\frac{\gamma_2^\mathsf{S}}{\hat{\gamma}_2}}\sum_{l=1}^{L}\frac{\beta_{2l}\tilde{\gamma}_2^{\kappa_l}}{\omega_2^{\kappa_l}}\gamma\left(\kappa_l,\omega_2\gamma_2^\mathsf{S}\right)\right).\label{eq-out-s1}
	\end{align}
	\hrulefill\vspace{-0.6cm}
\end{figure*}
The sum-rate OP $P_{o2}^\mathsf{S}$ can be also expressed as:
\begin{align}
&P_{o2}^\mathsf{S}=\Pr\left(C^\mathsf{S}\leq R_t^\mathsf{S}\right)=\Pr\left(\log_2\left(1+\gamma_1\right)\leq R_t^\mathsf{S}\right),\\
&=F_{\gamma_1}\left(\gamma_t^\mathsf{S}\right),\label{eq-s2}
\end{align}
where $\gamma_t^\mathsf{S}=2^{R_t^\mathsf{S}}-1$. By substituting \eqref{cdf} into \eqref{eq-s2}, $P_{o2}^\mathsf{S}$ is achieved as follows:
		\begin{align}
	P_{o2}^\mathsf{S}=\sum_{l=1}^{L}\frac{\beta_{1l}}{\zeta_1^{\kappa_l}}\gamma\left(\kappa_l,\frac{\gamma_t^\mathsf{S}\zeta_1}{\tilde{\gamma}_1}\right)-\mathrm{e}^{-\frac{\gamma_t^\mathsf{S}}{\hat{\gamma}_1}}\sum_{l=1}^{L}\frac{\beta_{1l}\tilde{\gamma}_1^{\kappa_l}}{\omega_1^{\kappa_l}}\gamma\left(\kappa_l,\omega_1\gamma_t^\mathsf{S}\right),\label{eq-out-s2}
\end{align}
Finally, by inserting $P_{o1}^\mathsf{S}$ and $P_{o2}^\mathsf{S}$ into \eqref{eq-out-sd}, $P_{o}^\mathsf{S}$ is obtained as \eqref{eq-out-s}. 
\begin{figure*}[t]
	\normalsize
	%\hrulefill
	\setcounter{equation}{32}
	\small\begin{align}\nonumber
		\hspace{-0.15cm}P_{o}^\mathsf{S}=&\rho\Bigg[1-\left(1-\sum_{l=1}^{L}\frac{\beta_{1l}}{\zeta_1^{\kappa_l}}\gamma\left(\kappa_l,\frac{\gamma_2^\mathsf{S}\zeta_1}{\tilde{\gamma}_1}\right)+\mathrm{e}^{-\frac{\gamma_2^\mathsf{S}}{\hat{\gamma}_1}}\sum_{l=1}^{L}\frac{\beta_{1l}\tilde{\gamma}_1^{\kappa_l}}{\omega_1^{\kappa_l}}\gamma\left(\kappa_l,\omega_1\gamma_2^\mathsf{S}\right)\right) \left(1-\sum_{l=1}^{L}\frac{\beta_{2l}}{\zeta_2^{\kappa_l}}\gamma\left(\kappa_l,\frac{\gamma_2^\mathsf{S}\zeta_2}{\tilde{\gamma}_2}\right)+\mathrm{e}^{-\frac{\gamma_2^\mathsf{S}}{\hat{\gamma}_2}}\sum_{l=1}^{L}\frac{\beta_{2l}\tilde{\gamma}_2^{\kappa_l}}{\omega_2^{\kappa_l}}\gamma\left(\kappa_l,\omega_2\gamma_2^\mathsf{S}\right)\right)\Bigg]\\
		&+\left(1-\rho\right)\left[\sum_{l=1}^{L}\frac{\beta_{1l}}{\zeta_1^{\kappa_l}}\gamma\left(\kappa_l,\frac{\gamma_t^\mathsf{S}\zeta_1}{\tilde{\gamma}_1}\right)-\mathrm{e}^{-\frac{\gamma_t^\mathsf{S}}{\hat{\gamma}_1}}\sum_{l=1}^{L}\frac{\beta_{1l}\tilde{\gamma}_1^{\kappa_l}}{\omega_1^{\kappa_l}}\gamma\left(\kappa_l,\omega_1\gamma_t^\mathsf{S}\right)\right].\label{eq-out-s}
	\end{align}
	\hrulefill\vspace{-0.65cm}
\end{figure*}
\begin{figure*}[!t]
	\centering
	\hspace{-0.2cm}\subfigure[Doubly dirty MAC]{%
		\includegraphics[width=0.32\textwidth]{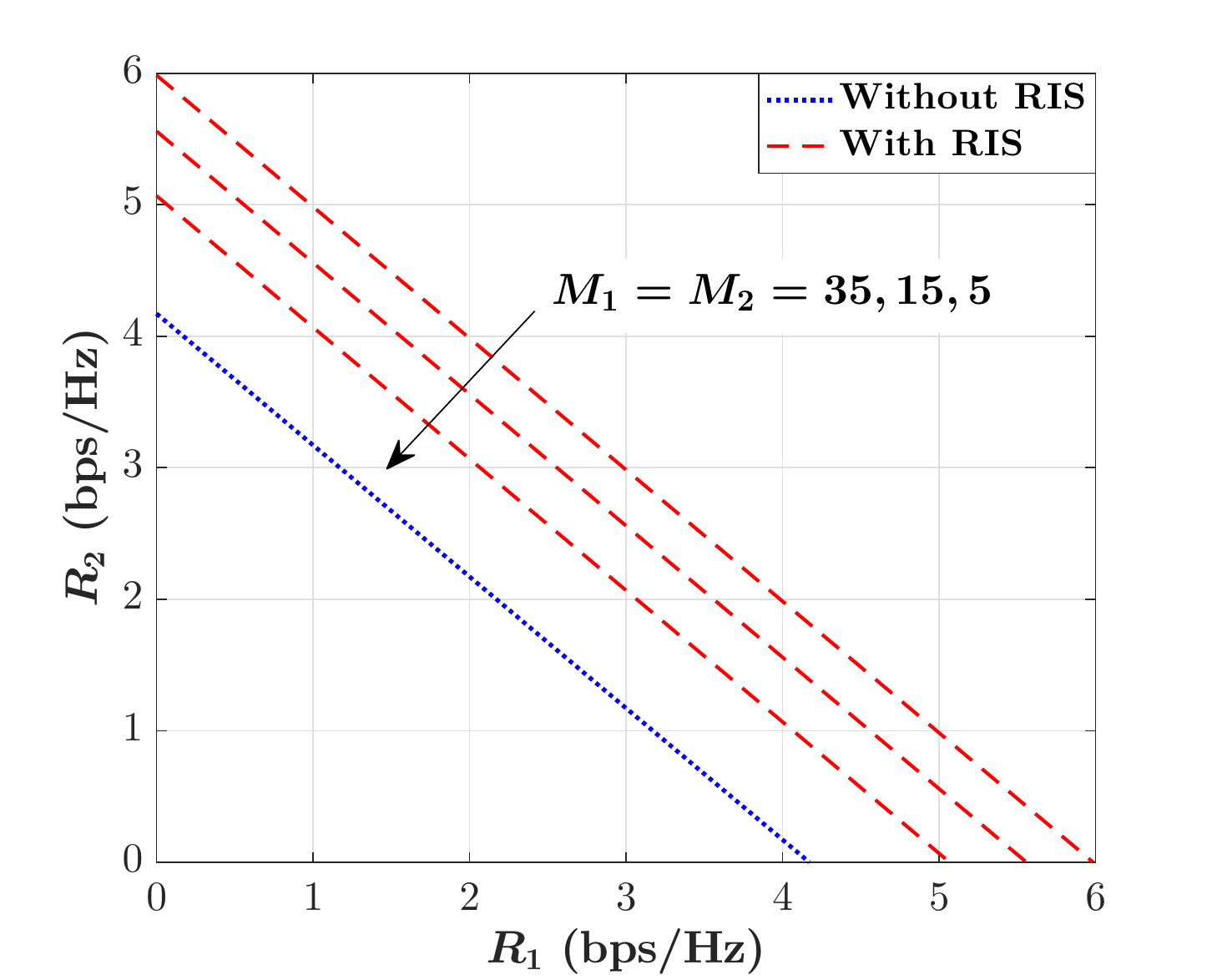}%
		\label{fig-cdd}%
	}\hspace{-0.3cm}%or more
	\subfigure[Single dirty MAC, $P_1>P_2$]{%
		\includegraphics[width=0.32\textwidth]{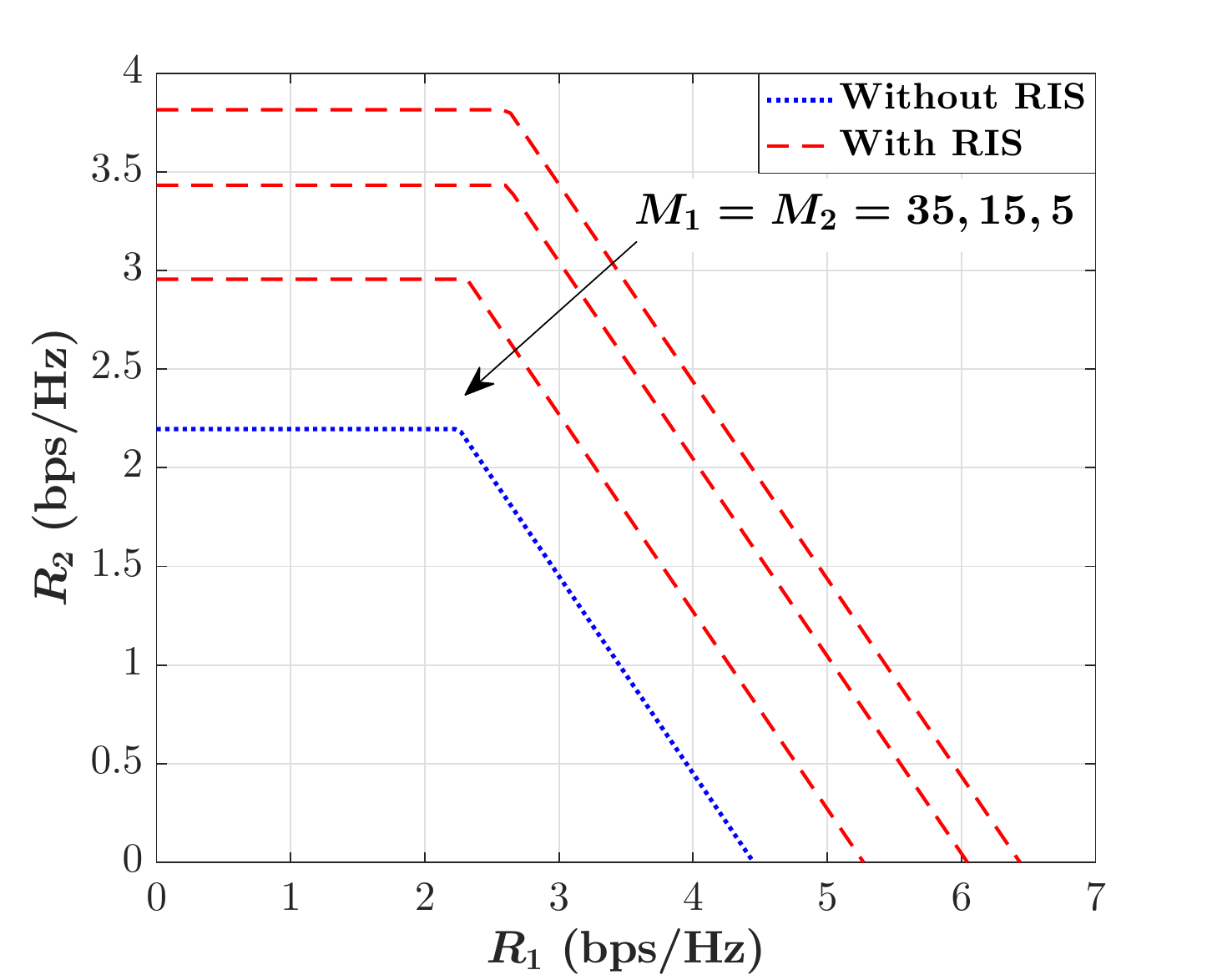}%
		\label{fig-csd1}%
	}\hspace{-0.3cm}%or more
	\subfigure[Single dirty MAC, $P_1\leq P_2$]{%
		\includegraphics[width=0.32\textwidth]{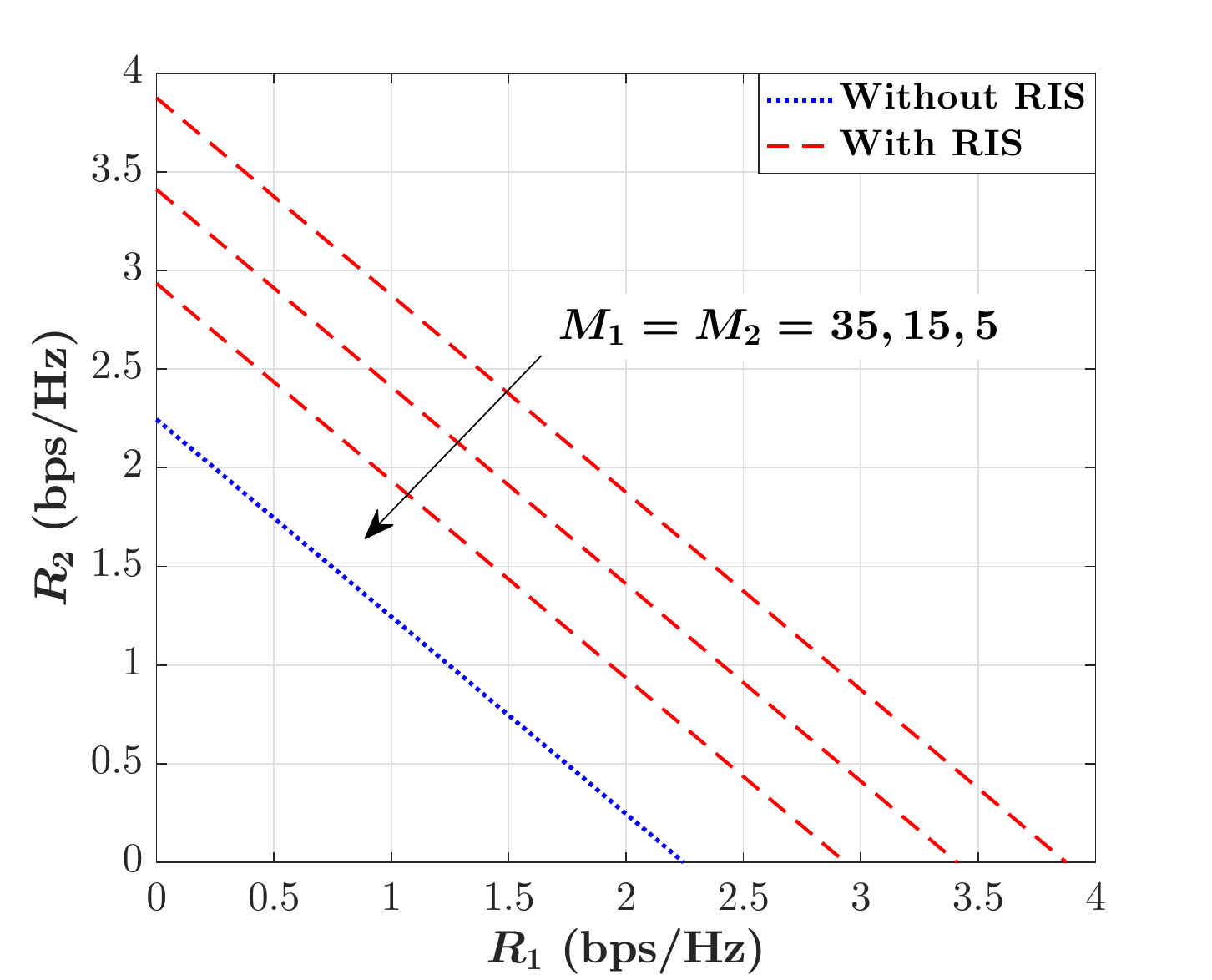}%
		\label{fig-csd2}%
	}\hspace{-0.3cm}%or more
	\caption{Capacity region for doubly/single dirty MAC under RIS deployment when $N=-10$dBm, $\check{d}_1=\check{d}_2=20$m, and (a) $P_1=P_2=60$dBm; (b) $P_1=45$dBm, $P_2=60$dBm; (c) $P_1=60$dBm, $P_2=45$dBm.}\vspace{-0.6cm}
	\label{fig-capacity}
\end{figure*}	
\begin{figure*}[!t]
	\centering
	\hspace{-0.2cm}\subfigure[Doubly dirty MAC]{%
		\includegraphics[width=0.32\textwidth]{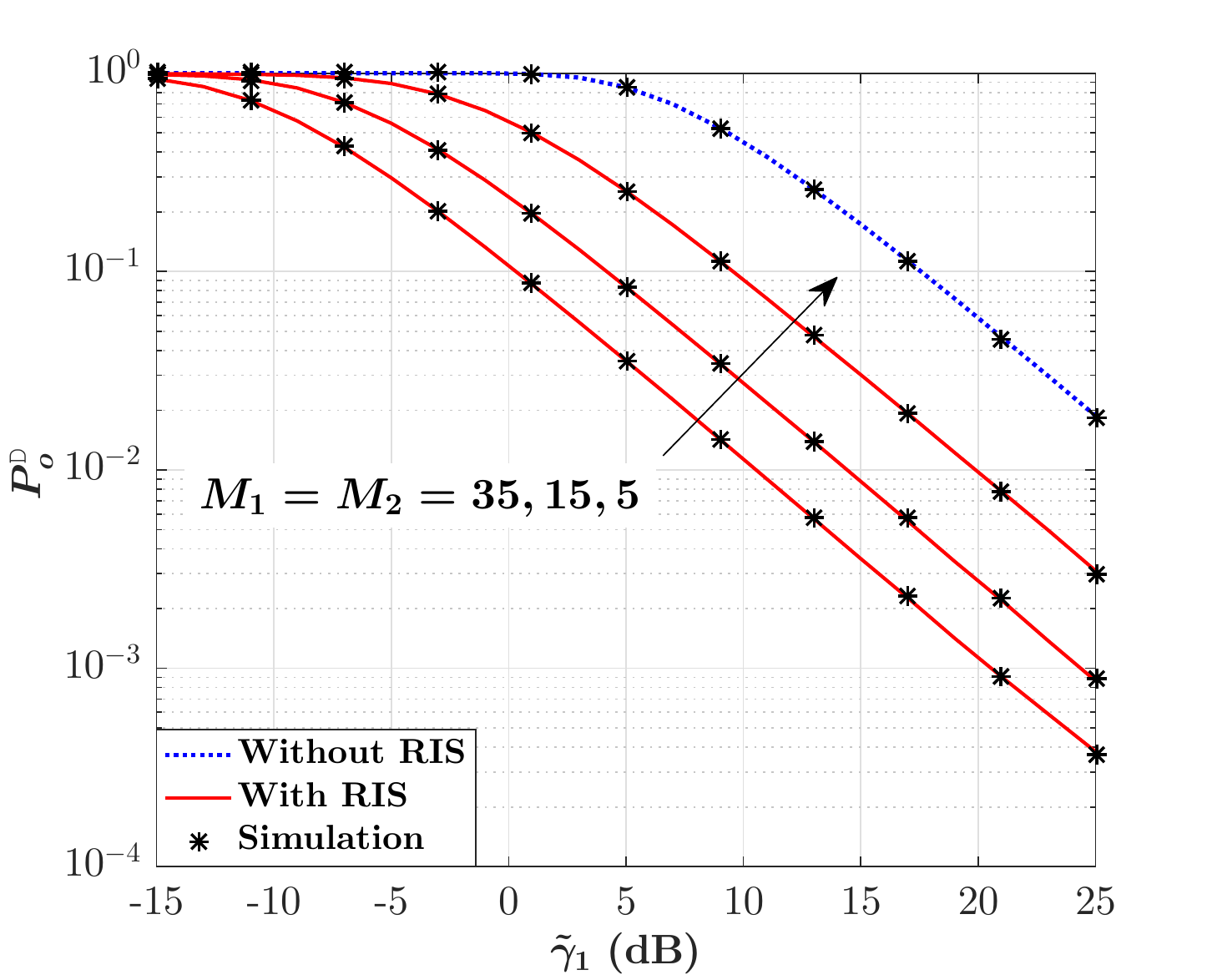}%
		\label{fig-outd}%
	}\hspace{-0.3cm}%or more
	\subfigure[Single dirty MAC, $P_1>P_2$]{%
		\includegraphics[width=0.32\textwidth]{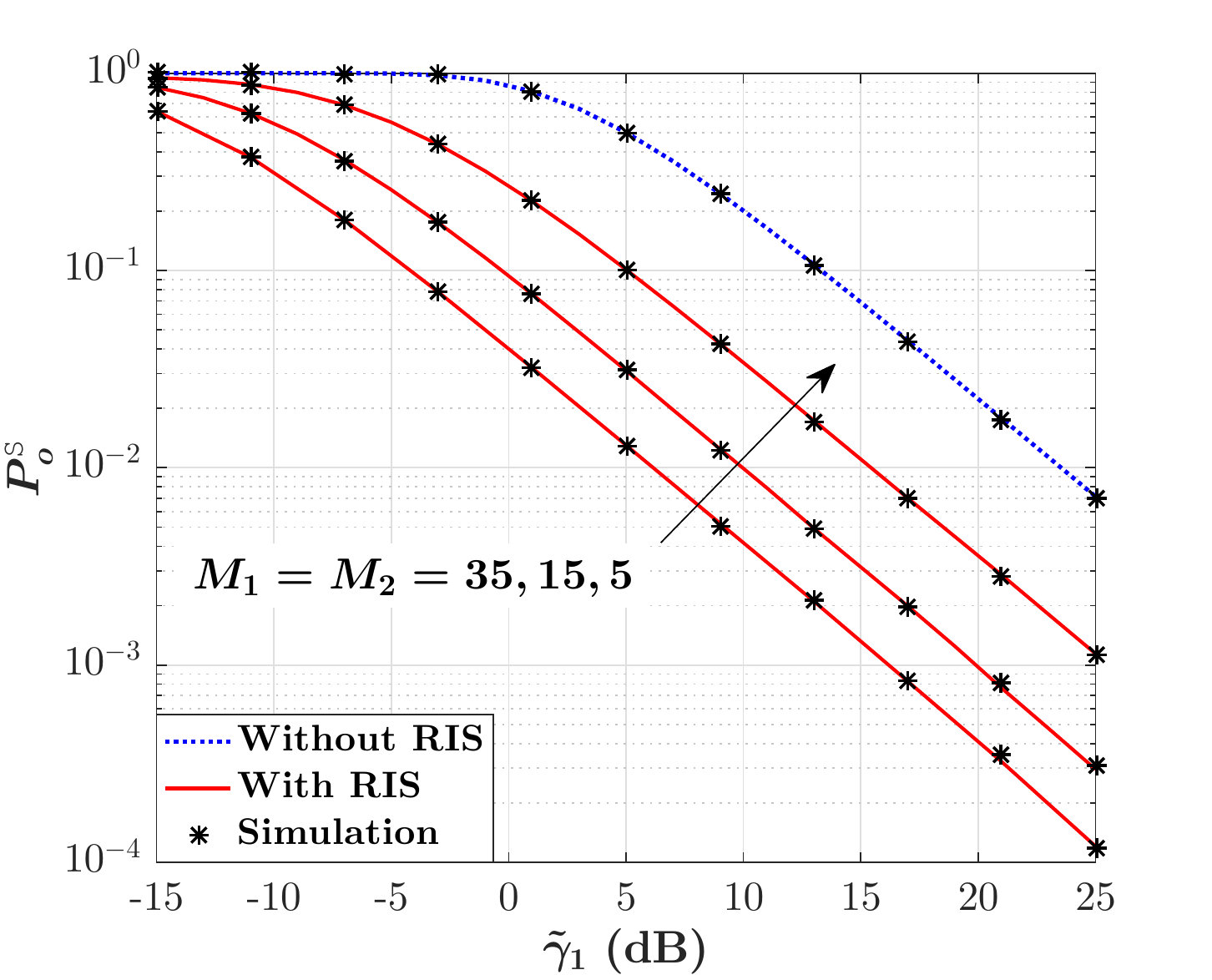}%
		\label{fig-outs1}%
	}\hspace{-0.3cm}%or more
	\subfigure[Single dirty MAC, $P_1\leq P_2$]{%
		\includegraphics[width=0.32\textwidth]{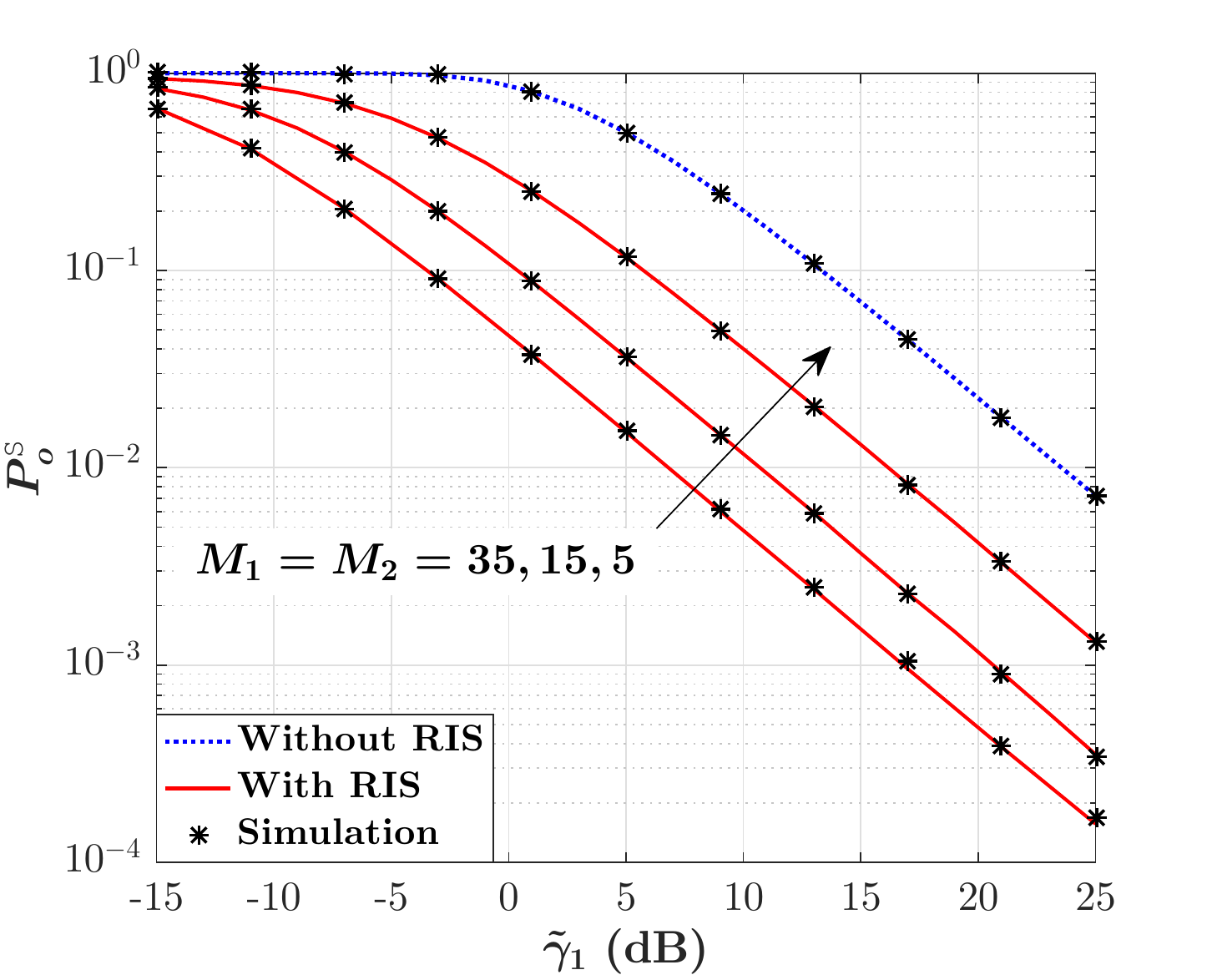}%
		\label{fig-outs2}%
	}\hspace{-0.3cm}%or more
	\caption{OP versus average SNR $\tilde{\gamma}_1$ for doubly/single dirty MAC under RIS deployment when $\tilde{\gamma}_2=\tilde{\gamma}_1$, $R_t^\mathsf{D}=1$bps/Hz, $R_t^\mathsf{S}=1$bps/Hz, $R_2^\mathsf{S}=0.5$bps/Hz, $N=10$dBm, and (a) $P_1=P_2=50$dBm; (b) $P_1=50$dBm, $P_2=40$dBm; (c) $P_1=40$dBm, $P_2=50$dBm.}\vspace{-0.5cm}
	\label{fig-outage}
\end{figure*}	\vspace{-0.5cm}		
\section{Numerical Results}
In this section, we present the numerical results to validate the theoretical expressions previously
derived, which are double-checked in all instances with Monte
Carlo (MC) simulations. For this purpose, we first suppose that the common receiver $r$ is located at $\left(0,0,6\right)$ under a three-dimensional coordinate system. We also consider that users $u_1$ and $u_2$ are placed at $\left(\check{d}_1,0,1\right)$ and $\left(-\check{d}_2,0,1\right)$, and the corresponding RISs are located at the $\left(\check{d}_1,0,2\right)$ and $\left(-\check{d}_2,0,2\right)$, respectively, where $\check{d}_i$ denotes the horizontal distance between $i$-th user and the receiver $r$. We set the path-loss exponents as $\hat{\alpha}=3$ for the direct channel $\hat{h}_i$ from user $u_i$ to the receiver $r$, $\bar{\alpha}=3$ for the channels $\vec{h}_i$ from user $u_i$ to its serving RIS, and $\tilde{\alpha}=3.5$ fro the reflected channels from the respective RIS to the receiver $r$. Fig. \ref{fig-capacity} shows the capacity region of doubly/single dirty MAC under Rayleigh fading channels with and without RIS deployment. In all three curves, it can be seen that considering RIS in both models provides a larger capacity region as compared with no RIS deployment since the RIS can improve the channel quality when phase processing is used. We can also see that increasing the number of reflecting meta-surface elements $M_i$ improves the capacity region for both doubly and single dirty MAC models. Fig. \ref{fig-outage} illustrates the behavior of the OP in terms of the average SNR $\tilde{\gamma}_1$ for doubly/single dirty MAC with and without RIS deployment, where the channels suffer from Rayleigh fading. It is clearly seen that applying the RIS improves the system performance in terms of the OP for both doubly and single dirty MAC system models in all curves since the RIS provides a higher received SNR under phase processing usage. As expected, we also observe that as the number of meta-surface elements $M_i$ grows, the OP performance significantly improves.\vspace{-0.3cm}
\section{Conclusion}
We analyzed the capacity region of the two-user MAC in the presence of non-casually known SI at the transmitters under RIS deployment. For each of the investigated scenarios, namely doubly dirty MAC where both users know the interferences and single dirty MAC where interference is only known for one of the users, results show that considering RIS in dirty MAC has constructive effects on the capacity region and the OP performance, i.e., provides a larger capacity region and a lower OP.  \vspace{-0.3cm}

\bibliographystyle{IEEEtran}
\bibliography{sample.bib}

\begin{thebibliography}{10}
\providecommand{\url}[1]{#1}
\csname url@samestyle\endcsname
\providecommand{\newblock}{\relax}
\providecommand{\bibinfo}[2]{#2}
\providecommand{\BIBentrySTDinterwordspacing}{\spaceskip=0pt\relax}
\providecommand{\BIBentryALTinterwordstretchfactor}{4}
\providecommand{\BIBentryALTinterwordspacing}{\spaceskip=\fontdimen2\font plus
\BIBentryALTinterwordstretchfactor\fontdimen3\font minus
  \fontdimen4\font\relax}
\providecommand{\BIBforeignlanguage}[2]{{%
\expandafter\ifx\csname l@#1\endcsname\relax
\typeout{** WARNING: IEEEtran.bst: No hyphenation pattern has been}%
\typeout{** loaded for the language `#1'. Using the pattern for}%
\typeout{** the default language instead.}%
\else
\language=\csname l@#1\endcsname
\fi
#2}}
\providecommand{\BIBdecl}{\relax}
\BIBdecl

\bibitem{basar2019wireless}
E.~Basar, M.~Di~Renzo, J.~De~Rosny, M.~Debbah, M.-S. Alouini, and R.~Zhang,
  ``Wireless communications through reconfigurable intelligent surfaces,''
  \emph{IEEE Access}, vol.~7, pp. 116\,753--116\,773, 2019.

\bibitem{jafar2006capacity}
S.~Jafar, ``Capacity with causal and noncausal side information: A unified
  view,'' \emph{IEEE Trans. Inf. Theory}, vol.~52, no.~12, pp. 5468--5474,
  2006.

\bibitem{philosof2011lattice}
T.~Philosof, R.~Zamir, U.~Erez, and A.~J. Khisti, ``Lattice strategies for the
  dirty multiple access channel,'' \emph{IEEE Trans. Inf. Theory}, vol.~57,
  no.~8, pp. 5006--5035, 2011.

\bibitem{ghadi2021role}
F.~R. Ghadi, G.~A. Hodtani, and F.~J. L{\'o}pez-Mart{\'\i}nez, ``{The role of
  correlation in the doubly dirty fading MAC with side information at the
  transmitters},'' \emph{IEEE Wirel. Commun. Lett.}, vol.~10, no.~9, pp.
  2070--2074, 2021.

\bibitem{ye2020joint}
J.~Ye, S.~Guo, and M.-S. Alouini, ``{Joint reflecting and precoding designs for
  SER minimization in reconfigurable intelligent surfaces assisted MIMO
  systems},'' \emph{IEEE Trans. Wireless Commun.}, vol.~19, no.~8, pp.
  5561--5574, 2020.

\bibitem{cui2019secure}
M.~Cui, G.~Zhang, and R.~Zhang, ``Secure wireless communication via intelligent
  reflecting surface,'' \emph{IEEE Wirel. Commun. Lett.}, vol.~8, no.~5, pp.
  1410--1414, 2019.

\bibitem{ai2021secure}
Y.~Ai, A.~Felipe, L.~Kong, M.~Cheffena, S.~Chatzinotas, and B.~Ottersten,
  ``Secure vehicular communications through reconfigurable intelligent
  surfaces,'' \emph{IEEE Trans. Veh. Technol.}, vol.~70, no.~7, pp. 7272--7276,
  2021.

\bibitem{yang2020secrecy}
L.~Yang, J.~Yang, W.~Xie, M.~O. Hasna, T.~Tsiftsis, and M.~Di~Renzo, ``{Secrecy
  performance analysis of RIS-aided wireless communication systems},''
  \emph{IEEE Trans. Veh. Technol.}, vol.~69, no.~10, pp. 12\,296--12\,300,
  2020.

\bibitem{mu2021joint}
X.~Mu, Y.~Liu, L.~Guo, J.~Lin, and R.~Schober, ``{Joint deployment and multiple
  access design for intelligent reflecting surface assisted networks},''
  \emph{IEEE Trans. Wireless Commun.}, vol.~20, no.~10, pp. 6648--6664, 2021.

\bibitem{zhou2020intelligent}
G.~Zhou, C.~Pan, H.~Ren, K.~Wang, and A.~Nallanathan, ``{Intelligent reflecting
  surface aided multigroup multicast MISO communication systems},'' \emph{IEEE
  Trans. Signal Process.}, vol.~68, pp. 3236--3251, 2020.

\bibitem{wei2021channel}
L.~Wei, C.~Huang, G.~C. Alexandropoulos, C.~Yuen, Z.~Zhang, and M.~Debbah,
  ``{Channel estimation for RIS-empowered multi-user MISO wireless
  communications},'' \emph{IEEE Trans. Commun.}, vol.~69, no.~6, pp.
  4144--4157, 2021.

\bibitem{gan2021ris}
X.~Gan, C.~Zhong, C.~Huang, and Z.~Zhang, ``{RIS-Assisted Multi-User MISO
  Communications Exploiting Statistical CSI},'' \emph{IEEE Trans. Commun.},
  vol.~69, no.~10, pp. 6781--6792, 2021.

\bibitem{xiao2021average}
G.~Xiao, T.~Yang, C.~Huang, X.~Wu, H.~Feng, and B.~Hu, ``{Average Rate
  Approximation and Maximization for RIS-Assisted Multi-user MISO System},''
  \emph{IEEE Wirel. Commun. Lett.}, 2021.

\bibitem{liu2021dynamic}
Y.~Liu, Y.~Zhang, X.~Zhao, S.~Geng, P.~Qin, and Z.~Zhou, ``{Dynamic-controlled
  RIS Assisted Multi-User MISO Downlink System: Joint Beamforming Design},''
  \emph{IEEE Trans. Green Commun. Netw.}, 2021.

\bibitem{zhang2020capacity1}
S.~Zhang and R.~Zhang, ``{Capacity characterization for intelligent reflecting
  surface aided MIMO communication},'' \emph{IEEE J. Sel. Areas Commun.},
  vol.~38, no.~8, pp. 1823--1838, 2020.

\bibitem{karasik2020beyond}
R.~Karasik, O.~Simeone, M.~Di~Renzo, and S.~S. Shitz, ``Beyond max-snr: Joint
  encoding for reconfigurable intelligent surfaces,'' in \emph{2020 IEEE Int.
  Symp. Inf. Theory (ISIT)}.\hskip 1em plus 0.5em minus 0.4em\relax IEEE, 2020,
  pp. 2965--2970.

\bibitem{zhang2021intelligent}
S.~Zhang and R.~Zhang, ``Intelligent reflecting surface aided multi-user
  communication: Capacity region and deployment strategy,'' \emph{IEEE Trans.
  Commun.}, vol.~69, no.~9, pp. 5790--5806, 2021.

\bibitem{proakis2001digital}
J.~G. Proakis and M.~Salehi, \emph{Digital communications}.\hskip 1em plus
  0.5em minus 0.4em\relax McGraw-hill New York, 2001, vol.~4.

\bibitem{atapattu2011mixture}
S.~Atapattu, C.~Tellambura, and H.~Jiang, ``A mixture gamma distribution to
  model the snr of wireless channels,'' \emph{IEEE Trans. Wireless Commun.},
  vol.~10, no.~12, pp. 4193--4203, 2011.

\end{thebibliography}
\end{document}